\documentclass[12pt]{article}

\usepackage{geometry}
\usepackage{cite}
\usepackage{paralist}
\usepackage{graphicx}
\usepackage{subfigure}
\usepackage{amssymb}
\usepackage{amsmath}
\allowdisplaybreaks[1]

\usepackage{amsthm}

\newtheorem{theorem}{Theorem}
\newtheorem{lemma}[theorem]{Lemma}
\newtheorem{corollary}[theorem]{Corollary}

\makeatletter
\def\@endtheorem{\endtrivlist}
\makeatother

\newcounter{brule}
\newenvironment{brule}{\refstepcounter{brule}\par\smallskip\noindent
\textbf{(B\arabic{brule})}\quad}{}
\newcommand{\currentrule}{B\arabic{brule}}

\newcommand{\Nnext}[1]{N_{\mathrm{next}}(#1)}
\newcommand{\Nsame}[1]{N_{\mathrm{same}}(#1)}
\newcommand{\Nprev}[1]{N_{\mathrm{prev}}(#1)}

\begin{document}

\title{Cluster deletion revisited}
\author{Dekel Tsur%
\thanks{Ben-Gurion University of the Negev.
Email: \texttt{dekelts@cs.bgu.ac.il}}}
\date{}
\maketitle

\begin{abstract}
In the \textsc{Cluster Deletion} problem the input is
a graph $G$ and an integer $k$, and the goal is to decide whether there
is a set of at most $k$ edges whose removal from $G$ results a graph in
which every connected component is a clique.
In this paper we give an algorithm for \textsc{Cluster Deletion} whose running
time is $O^*(1.404^k)$.
\end{abstract}


\section{Introduction}
A graph $G$ is called a \emph{cluster graph} if every connected component of
$G$ is a clique.
In the \textsc{Cluster Deletion} problem the input is
a graph $G$ and an integer $k$, and the goal is to decide whether there
is a set of at most $k$ edges whose removal from $G$ results a clique
graph.

A graph $G$ is a cluster graph if and only if there is no induced $P_3$ in $G$,
where $P_3$ is a path on 3 vertices.
Therefore, there is a simple $O^*(2^k)$-time branching algorithm for
\textsc{Cluster Deletion}: If the current graph is not a cluster graph, find an
induced $P_3$ in the graph and create two new instances by removing each
edge of the path.
Faster algorithms for \textsc{Cluster Deletion} were given
in~\cite{gramm2005graph,gramm2004automated,damaschke2009bounded,bocker2011even}.
First, an $O^*(1.77^k)$-time algorithm was given in~\cite{gramm2005graph}.
An $O^*(1.53^k)$-time algorithm was given in~\cite{gramm2005graph}, and
an $O^*(1.47^k)$-time algorithm was given in~\cite{damaschke2009bounded}.
Finally, B{\"o}cker and Damaschke~\cite{bocker2011even} gave an algorithm
with a claimed running time of $O^*(1.415^k)$.

In this paper we show that there is an error in the analysis of the algorithm
of B{\"o}cker and Damaschke. 
We give a corrected analysis that shows that the running time of the algorithm
is $O^*(1.415^k)$ as claimed in~\cite{bocker2011even}.
Additionally, we give an algorithm for \textsc{Cluster Deletion} whose
running time is $O^*(1.404^k)$.

\section{Preliminaries}
For set of vertices $S$ in a graph $G$, $G[S]$ is the subgraph
of $G$ induced by $S$ (namely, $G[S]=(S,E\cap (S\times S))$).
For a set of edges $F$, $G-F$ is the graph obtained from $G$ by deleting
the edges of $F$.
For a vertex $v$, $N(v)$ is the set of vertices that are adjacent to $v$.

Let $P_3$ denote a path on 3 vertices, and $C_4$ denote a chordless cycle on 4
vertices.


A graph $G$ is an \emph{$\alpha$-almost clique} if there is a set
$X \subseteq V(G)$ of size at most $\alpha$ such that $G[V(G)\setminus X]$
is a clique.

\begin{lemma}[Damaschke~\cite{damaschke2009bounded}]\label{lem:almost-clique}
Let $\alpha$ be a constant.
There is a polynomial time algorithm that given an $\alpha$-almost clique $G$,
finds a set of edges $S$ of minimum size such that $G-S$ is a cluster graph.
\end{lemma}


\section{The algorithm of B{\"o}cker and Damaschke}\label{sec:simple}
In this section we describe the algorithm of
B{\"o}cker and Damaschke~\cite{bocker2011even} and give a corrected analysis of
the algorithm.
The algorithm is a branching algorithm.
When we say that the algorithm \emph{branches} on sets $S_1,\ldots,S_p$, we mean
that the algorithm is called recursively on the instances
$(G-S_1,k-|S_1|),\ldots,(G-S_p,k-|S_p|)$.

The algorithm uses the following branching rules.
\begin{brule}
Suppose that $P$ is an induced path in $G$ with at least 7 vertices,
and let $e_1,e_2,\ldots,e_p$ be the edges of $P$ (in order).
Branch on $\{e_1,e_3,\ldots\}$ and $\{e_2,e_4,\ldots\}$.\label{rule:path}
\end{brule}

We note that Rule~(\currentrule) is a simplified version of the corresponding
rule given in~\cite{bocker2011even}.
To show the safeness of this rule, suppose that $(G,k)$ is a yes instance
and let $S$ be a solution of $(G,k)$ (namely, $G-S$ is a cluster graph and
$|S| \leq k$).
Since $V(e_1) \cup V(e_2)$ induces a $P_3$ in $G$, $S$ contains at least one
edge from $e_1,e_2$.
Suppose that $e_1 \in S$.
The set $V(e_2)\cup V(e_3)$ induces a $P_3$ in $G$, so either $e_2 \in S$ or
$e_3 \in S$.
Therefore, the set $S_2 = (S\setminus \{e_2\}) \cup \{e_3\}$
is also a solution of $(G,k)$
(since the deletion of $e_3$ from $G$ does not generate new induced $P_3$'s).
By repeating this process we obtain a solution $S'$ such that
$\{e_1,e_3,\ldots\} \subseteq S'$.
Similarly, if $e_2 \in S$ we can obtain a solution $S'$ such that
$\{e_2,e_4,\ldots\} \subseteq S'$.
Therefore, Rule~(\currentrule) is safe.

The branching vector of Rule~(\currentrule) is at least $(3,3)$
and the branching number is less than 1.26.

For an edge $e$, let $F_e$ be a set containing all edges $e'$ such that
$V(e) \cup V(e')$ induces a $P_3$.

\begin{brule}
If there is an edge $e$ for which $|F_e| \geq 4$,
branch on $\{e\}$ and $F_e$.\label{rule:14}
\end{brule}

The branching vector of Rule~(\currentrule) is at least $(1,4)$
and the branching number is less than 1.381.

\begin{brule}
Let $v_1,v_2,v_3,v_4,v_1$ be an induced $C_4$ in $G$.
Branch on $\{(v_1,v_2),(v_3,v_4)\}$ and
$\{(v_2,v_3),(v_4,v_1)\}$.\label{rule:c4}
\end{brule}

The branching vector of Rule~(\currentrule) is $(2,2)$
and the branching number is less than 1.415.

\refstepcounter{brule}
\label{rule:P3}
\refstepcounter{brule}
\label{rule:P3new}
The algorithm applies Rules (B1)--(B\ref{rule:c4}) until none of these rule is
applicable.
If $G$ is a graph on which Rules (B1)--(B\ref{rule:c4}) are not applicable
and $G$ is not a cluster graph, the algorithm applies a new branching rule,
denoted (B\ref{rule:P3}), which is described below.

In Rule~(B\ref{rule:P3}), the algorithm finds an induced $P_3$ $u,v,w$.
To describe the rule, we define the following sets.
Let $A = \{u,v,w\}$.
Let $B$ the set of all vertices not in $A$ with one or two neighbors in $A$.
Let $B_u$ (resp., $B_w$) be the set of all vertices $b\in B$ such that
$b$ is adjacent to exactly one vertex from $u,v$ (resp., from $w,v$).
Clearly, $B = B_u \cup B_w$.
Let $C$ be the set of all vertices with three neighbors in $A$.
Let $D$ be the set of all vertices not in $A \cup B \cup C$ that have at least
one neighbor in $C$.

For $i \geq 1$, let $B_i$ be a set containing every vertex $x$
not in $A \cup B \cup C \cup D$ such that the minimum distance between $x$ and
a vertex in $B$ is exactly $i$. Additionally, denote $B_0 = B$ and $B_{-1} = A$.
For a vertex $x \in B_i$, let $\Nnext{x}$, $\Nsame{x}$, and $\Nprev{x}$ denote
the sets of neighbors of $x$ in $B_{i+1}$, $B_i$, and $B_{i-1}$, respectively.
Note that by definition, $\Nprev{x} \neq \emptyset$ for every
$x \in \bigcup_{i \geq 0} B_i$.
For $i \geq 0$, let $E_i$ be the set of all edges with one endpoint in $B_i$
and the other endpoint in $B_{i+1}$.

The following lemmas are proved in~\cite{bocker2011even}.
\begin{lemma}\label{lem:Bu}
If $G$ is a graph in which Rule~(B\ref{rule:14}) is not applicable
then $|B_u| \leq 2$ and $|B_w| \leq 2$.
\end{lemma}

\begin{lemma}\label{lem:C}
If $G$ is a graph in which Rule~(B\ref{rule:c4}) is not applicable
then $C$ is a clique.
\end{lemma}

\begin{lemma}\label{lem:D}
Let $G$ be a graph in which Rule~(B\ref{rule:14}) is not applicable
and $C$ is clique.
Then, $|D| \leq 3$.
Additionally, for every $x\in D$, $C \subseteq N(x) \subseteq B \cup C \cup D$.
\end{lemma}

We now give several lemmas which will be used in the analysis of
Rule~(B\ref{rule:P3}).
\begin{lemma}\label{lem:Nnext-2}
If $G$ is a graph in which Rule~(B\ref{rule:14}) cannot be applied,
$|\Nnext{x}| \leq 2$ for every $x \in \bigcup_{i \geq 0} B_i$.
\end{lemma}
\begin{proof}
We first consider the case $x \notin B_0$.
Let $y \in \Nprev{x}$ and $z\in \Nprev{y}$.
The set $F_{(x,y)}$ contains the edge $(y,z)$ and the edge
$(x,x')$ for every $x' \in \Nnext{x}$.
Therefore, $|\Nnext{x}| \leq |F_{(x,y)}|-1 \leq 2$.

We now consider the case $x \in B_0 = B$.
By the definition of $B$, there are vertices $y,z \in A$ such that
$y$ is adjacent to $x$ and $z$, and $z$ is not adjacent to $x$.
Therefore, the set $F_{(x,y)}$ contains the edge $(y,z)$ and the edge
$(x,x')$ for every $x' \in \Nnext{x}$.
We obtain again that $|\Nnext{x}| \leq 2$.
\end{proof}

By Lemma~\ref{lem:Nnext-2} we have that $|B_1| \leq 2|B| \leq 8$ and
$|E_0| \leq 8$.

Let $j$ be the minimum index such that there is an edge $e \in E_j$
for which $|F_e| \geq 3$
(we note that while $G$ does not have an edge $e$ with $|F_e| > 3$,
we will later use this definition on graphs which can have such edges).
If no such index exists, $j = \infty$.
\begin{lemma}\label{lem:Nnext}
If $j\geq 1$, $|\Nnext{x}| \leq 1$ for every $x \in B_1 \cup \cdots \cup B_j$.
\end{lemma}
\begin{proof}
Let $y \in \Nprev{x}$ and $z\in \Nprev{v}$.
As in the proof of Lemma~\ref{lem:Nnext-2}, $|\Nnext{x}| \leq |F_{(x,y)}|-1$.
Since the edge $(x,y)$ is in $E_0 \cup \cdots \cup E_{j-1}$,
by the definition of $j$, $|F_{(x,y)}| \leq 2$ and the lemma follows.
\end{proof}
\begin{lemma}\label{lem:Nsame}
If $j\geq 2$, $|\Nsame{x}| \leq 1$ for every $x \in B_2 \cup \cdots \cup B_j$.
Additionally, if $|\Nsame{x}| = 1$ then $\Nnext{x} = \emptyset$.
\end{lemma}
\begin{proof}
Let $y \in \Nprev{x}$ and $z\in \Nprev{v}$.
Since $y \in B_1 \cup\cdots\cup B_{j-1}$, by Lemma~\ref{lem:Nnext} we have that
$\Nnext{y} = \{x\}$.
Therefore, for every $x' \in \Nsame{x}$ we have that $(x,x') \in F_{(x,y)}$.
Since $F_{(x,y)}$ also contains every $x' \in \Nnext{x}$, we conclude that
`$|\Nsame{x}|+|\Nnext{x}| \leq |F_{(x,y)}|-1 \leq 1$, and the lemma follows.
\end{proof}
\begin{lemma}\label{lem:Nprev}
If $j\geq 3$, $|\Nprev{x}| \leq 2$ for every $x \in B_3 \cup \cdots \cup B_j$.
Additionally, if $|\Nprev{x}| = 2$ then $\Nnext{x} = \Nsame{x} = \emptyset$.
\end{lemma}
\begin{lemma}\label{lem:Nprev-2}
If $j\geq 2$, if $x \in B_2$ and $\Nnext{x} \neq \emptyset$ then
$|\Nprev{x}| \leq 2$.
Additionally, let $x'$ be the unique vertex in $\Nnext{x}$.
Then, $N(x') = \{x\}$.
\end{lemma}
The proofs of Lemma~\ref{lem:Nprev} and Lemma~\ref{lem:Nprev-2} are similar to
the proof of Lemma~\ref{lem:Nsame} and were thus omitted.

From the lemmas above we obtain the following corollaries.
\begin{corollary}\label{cor:Ei}
$|E_j| \leq |E_{j-1}| \leq \cdots \leq|E_0|$.
\end{corollary}
\begin{corollary}\label{cor:paths}
If $j \geq 2$, $G[B_2 \cup \cdots \cup B_j]$ is a collection of at most
$|E_1|$ disjoint paths.
\end{corollary}

When we consider a subgraph $G'$ of $G$, we use $G'$ in  superscript to refer
to a set (or an integer) defined for $G'$.
For example, the set of all vertices $b$ such that $b \notin A$ and $b$ is
adjacent in $G'$ to one or two vertices in $A$ is denoted $B^{G'}$.

We now describe Rule~(B\ref{rule:P3}).
Recall that the algorithm first finds an induced $P_3$ $u,v,w$.
The main idea is to disconnect $u,v,w$ from the rest of the graph.
If $j \neq \infty$, the algorithm picks an arbitrary edge $e \in E_j$ such that
$|F_e| \geq 3$.
Then, the algorithm branches on $\{e\}$ and $F_e$.
Namely, the algorithm builds two instances
$(G_1,k_1) = (G-\{e\},k-1)$ and $(G_2,k_2) = (G-F_e,k-|F_e|)$.
This process is repeated on each of these two instances:
Let $(G_i,k_i)$ be one of the two instances.
If $j^{G_i} \neq \emptyset$, the algorithm picks an arbitrary edge
$e_i \in E_{j^{G_i}}^{G_i}$ such that $|F_{e_i}^{G_i}| \geq 3$ and creates
two new instances: $(G_{i1},k_{i1}) = (G_i-\{e_i\},k_i-1)$ and
$(G_{i2},k_{i2}) = (G_i-F_{e_i}^{G_i},k_i-|F_{e_i}^{G_i}|)$.
This is repeated until every instance $(G_{i_1\cdots i_p},k_{i_1\cdots i_p})$
satisfies $j^{G_{i_1\cdots i_p}} = \infty$.
An instance $(G_{i_1\cdots i_p},k_{i_1\cdots i_p})$ with
$j^{G_{i_1\cdots i_p}} = \infty$ generated by Rule~(B\ref{rule:P3}) will be
called a \emph{first stage instance}.

Next, on each first stage instance
$(G',k') = (G_{i_1\cdots i_p},k_{i_1\cdots i_p})$,
the algorithm repeatedly applies Rule~(B\ref{rule:path}) as follows.
Since $j^{G'} = \infty$, we have by Corollary~\ref{cor:paths} that
$G'[\bigcup_{i\geq 2} B_i^{G'}]$ is a collection of at most $|E_1^{G'}|$
disjoint paths.
On each path of these paths that has at least 7 vertices,
the algorithm applies Rule~(B\ref{rule:path}).
If there are $q$ such paths, this generates $2^q$ instances from $(G',k')$.
These instances will be called \emph{second stage instances}.

Let $(G'',k'')$ be a second stage instance and let $(G',k')$ be the first
stage instance from which $(G'',k'')$ was generated.
Note that $G''[\bigcup_{i\geq 2} B_i^{G''}]$ is a collection of at most
$|E_1^{G''}| = |E_1^{G'}|$ disjoint paths, where each path contains at most 6
vertices.
Let $H$ be the connected component of $v$ in $G''$.
We will later show that $G''$ is an $O(1)$-almost clique.
The algorithm uses the algorithm of Lemma~\ref{lem:almost-clique} to find
a set $S^{G''} \subseteq E(H)$ of minimum size such that $H-S^{G''}$ is a
cluster graph.
Then, the algorithm makes a recursive call on $(G''-S^{G''},k''-|S^{G''}|)$.
Note that the size of $S$ is at least 1 since $u,v,w$ is an induced $P_3$
in $H$.
This is repeated for every second stage instance.
The instances $(G''-S^{G''},k''-|S^{G''}|)$ will be called
\emph{third stage instances}.

We now show that $H$ is an $O(1)$-almost clique.
By Lemma~\ref{lem:D}, $V(H) = A \cup C \cup D \cup \bigcup_{i\geq 0} B_i^{G''}$.
We have that $|A| = 3$ and $|D| \leq 3$ (Lemma~\ref{lem:D}).
Additionally, $B_0^{G''} \cup B_1^{G''} \subseteq B_0 \cup B_1$
and therefore $|B_0^{G''}|+|B_1^{G''}| \leq |B_0|+|B_1| \leq 12$.
Moreover, $|\bigcup_{i\geq 2} B_i^{G''}| \leq 6|E_1^{G''}| = 6|E_1^{G'}|$ and
$|E_1^{G'}| = O(1)$ (since every edge in $E_1^{G'}$ is an edge in
$G[B\cup B_1]$).
Therefore, by Lemma~\ref{lem:C}, $H$ is an $O(1)$-almost clique.

We note that in the analysis of Rule~(B\ref{rule:P3}) in~\cite{bocker2011even},
it is claimed that if $(G_{i_1\cdots i_p},k_{i_1\cdots i_p})$
is a first stage instance then $p \leq 4$.
However, this is not true.
Suppose that $B = \{b_1,b_2,b_3,b_4\}$, $B_1 = \{c_1,c_2,c_3,c_4\}$,
$B_2 = \{d_1,\ldots,d_6\}$,
$N(b_1) = \{u,b_2,c_1,c_2\}$,
$N(b_2) = \{u,b_1,c_1,c_2\}$,
$N(b_3) = \{w,b_4,c_3,c_4\}$,
$N(b_4) = \{w,b_3,c_4,c_4\}$,
$N(c_1) = \{b_1,b_2,c_2,d_1,d_2\}$,
$N(c_2) = \{b_1,b_2,c_1,d_3\}$,
$N(c_3) = \{b_3,b_4,c_4,d_4,d_5\}$, and
$N(c_4) = \{b_3,b_4,c_3,d_6\}$.
In this graph, Rule~(B\ref{rule:P3}) generates an instance
$(G_{11111111},k_{11111111})$ by deleting one by one all 8 edges between
$B$ and $B_1$.

We now bound the branching number of Rule~(B\ref{rule:P3}).
We consider several cases.
In the first case, suppose that $j \geq 2$ and $j \neq \infty$.
Recall that the algorithm picks an edge $e = (x,y) \in E_j$ and generates
the instances $(G_1,k_1) = (G-\{e\},k-1)$ and $(G_2,k_2) = (G-F_e,k-|F_e|)$.
Clearly, $|E_j^{G_1}| = |E_j|-1$.
We now show that $|E_j^{G_2}| \leq |E_j|-1$.
Since $G_2$ is obtained from $G$ by erasing the edges in $F_e$, and all the
edges in $E_{j-1}$ that are incident on $x$ are in $F_e$, we have that
$x \notin B_j^{G_2}$.
Therefore, $(x,y) \notin E_j^{G_2}$.
We now claim that $E_j^{G_2}$ does not contain an edge that is not in $E_j$.
Suppose conversely that $e'$ is such an edge.
Then $x$ must be in $B_{j+1}^{G_2}$ and $e' = (x',x)$ for some
$x' \in B_j^{G_2}$.
Therefore, in $G$ we have $x' \in \Nsame{x}$.
This contradicts Lemma~\ref{lem:Nsame}.
Thus, $E_j^{G_2}$ does not contain edges that are not in $E_j$.
It follows that $|E_j^{G_2}| \leq |E_j|-1$.
By Corollary~\ref{cor:Ei}, $|E_{j^{G_i}}^{G_i}| \leq |E_j^{G_i}| \leq |E_j|-1$
for $i=1,2$.
Using the same arguments, for every first stage instance
$(G',k') = (G_{i_1\cdots i_p},k_{i_1 \cdots i_p})$ we have that
$|E_{j^{G'}}^{G'}| \leq |E_j|-p$.
This implies that $p \leq |E_j|$.
Therefore, the number of first stage instances generated by
Rule~(B\ref{rule:P3}) is at most $2^{|E_j|}$.

Suppose for example that $|E_j| = 2$.
In the worst case, Rule~(B\ref{rule:P3}) generates four first stage instances:
$(G_{11},k_{11})$, $(G_{12},k_{12})$, $(G_{21},k_{21})$, and $(G_{22},k_{22})$.
Additionally, $k_{11} = k-1$, $k_{12} = k_{21} = k-4$, and $k_{22} = k-6$.
Suppose that Rule~(B\ref{rule:path}) is not applied on any of the four instances
above.
Then, the algorithm generates a third stage instance
$(G_{i_1i_2}-S^{G_{i_1i_2}},k_{i_1i_2}-|S^{G_{i_1i_2}}|)$ from each
first stage instance $(G_{i_1i_2},k_{i_1i_2})$.
Since $|S^{G_{i_1i_2}}| \geq 1$ for every $i_1,i_2$,
we obtain that
$k_{11}-S^{G_{11}} \leq k-3$,
$k_{12}-S^{G_{12}} \leq k-5$,
$k_{21}-S^{G_{21}} \leq k-5$,
$k_{22}-S^{G_{22}} \leq k-7$.
In other words, the branching vector of Rule~(B\ref{rule:P3}) in this case is
at least $(3,5,5,7)$.

To analyze the case in which Rule~(B\ref{rule:path}) is applied (at least once)
on at least one of the four first stage instances
note that for the sake of the analysis, we can assume for the sake of the
analysis that the applications of Rule~(B\ref{rule:path}) are done after the
application of Rule~(B\ref{rule:P3}).
For example, suppose that Rule~(B\ref{rule:path}) is applied only on the
instance $(G_{11},k_{11})$ and it is applied once on this instance,
generating instances $(G'_{11},k'_{11})$ and $(G''_{11},k''_{11})$, where
$k'_{11} = k''_{11} = k_{11} - 3$.
Therefore, the branching vector of Rule~(B\ref{rule:P3}) in this case is at
least $(6,6,5,5,7)$.
However, we can assume for the analysis that Rule~(B\ref{rule:P3}) generates
only four third stage instances, namely the instances
$(G_{i_1i_2}-S^{G_{i_1i_2}},k_{i_1i_2}-|S^{G_{i_1i_2}}|)$,
and then the algorithm applies Rule~(B\ref{rule:path}) on
$(G_{11}-S^{G_{11}},k_{11}-|S^{G_{11}}|)$.
The branching vectors for these two rules are $(3,5,5,7)$ and $(3,3)$,
respectively.

For a general value of $|E_j|$, we have that in the worst case
Rule~(B\ref{rule:P3}) generates $2^{|E_j|}$ third stage instances
(as discussed above, we can assume that Rule~(B\ref{rule:path}) was not applied
on the first stage instances).
The branching vector is at least $R(|E_j|)$, where
$R(p)$ is a vector of length $2^p$ in which the value $p+1+2i$ appears
$\binom{p}{i}$ times, for $i = 0,\ldots,p$.
We note that this bound is not good enough for our purpose.
Even for $|E_j| = 5$, we have that the branching vector $R(5)$ has branching
number of approximately 1.406 and we need a branching number less than 1.404.
The solution is to give better bounds on the sizes of the sets $S_{G''}$.
This will be discussed below.

Now consider the case $j = 1$.
In this case $|E_1^{G_1}| = |E_1|-1$ as in the first case.
However, we now can have $|E_1^{G_2}| \geq |E_1|$.
This occurs if $\Nsame{x} \cap \Nprev{y} \neq \emptyset$.
In this case, $x$ belongs to $B^{G_2}_{2}$, and for every
$x' \in \Nsame{x} \cap \Nprev{y}$, the edge $(x,x')$ is in $E_1^{G_2}$ and not
in $E_j$.
Note that in this case we have that $x$ and $y$ are not adjacent in $G_2$ to
vertices in $B_2^{G_2} \cup B_3^{G_2}$ (If $z$ is adjacent in $G_2$ to $x$ or
$y$ then $z$ must be adjacent to both $x$ and $y$, otherwise the edge
between $z$ and $x$ or $y$ is in $F_{(x,y)}$.
$z$ is also adjacent to $x$ and $y$ in $G$.
Therefore, $z$ is in $B_1$ or $B_2$.
By Lemma~\ref{lem:Nnext}, $z \notin B_2$, so $z\in B_1$).
Therefore, we can ignore the edges $(x',x)$ and $(x',y)$.
Formally, if the case above occurs, we \emph{mark} the edges
$(x',x)$ and $(x',y)$ for every $x' \in \Nsame{x} \cap \Nprev{y}$.
We now change the definition of $E_1$ to include all the unmarked edges
with one endpoint in $B_1$ and the other endpoint in $B_2$.
For this new definition we have $|E_1^{G_2}| \leq |E_1|-1$.
Additionally, Corollary~\ref{cor:Ei} remains true.
Therefore, the analysis of the case $j=1$ is the same as the analysis 
for the case $j \geq 2$.
That is, for a specific value of $|E_1|$, the worst branching vector
is at least $R(|E_1|)$.

We now consider the case $j=0$.
To handle this case
(and to get a better bound on the branching number for the case $j\geq 1$),
we use a Python script.
The script goes over possible cases for the graph
$G[A\cup B\cup B_1 \cup B_2]$ and for each case it computes a branching vector
that gives an upper bound on the branching number for this case.
Formally,
a \emph{configuration graph} is a graph $J$ whose vertices are
partitioned into 4 sets:
(1) A set $A^J$ that contains 3 vertices $u',v',w'$ which form a $P_3$.
(2) A set $B^J$ that contains vertices that are adjacent to 1 or 2 vertices of
$A^J$.
(3) A set $B_1^J$ that contains vertices that are adjacent to at least one
vertex in $B^J$ and are not adjacent to vertices in $A^J$.
(4) A set $B_2^J$ such that every vertex in $B_2^J$ has exactly one neighbor
and this neighbor is in $B_1^J$.
Additionally, $|F_e^J| \leq 3$ for every edge $e$ with at least one endpoint in
$A^J \cup B^J$.
We note that the restriction on the degree of the vertices in $B_2^J$ is
required in order to restrict the number of configuration graphs.

Let $G$ be a graph with an induced $P_3$ $u,v,w$.
We say that $G$ \emph{matches} a configuration graph $J$ if there is a
bijection $\phi \colon A \cup B \cup B_1 \to A^J \cup B^J \cup B_1^J$ such that
(1) $\phi$ is an isomorphism between $G[A\cup B \cup B_1]$ and
$J[A^J\cup B^J\cup B_1^J]$.
(2) $\phi$ maps $u,v,w$ to $u',v',w'$, respectively.
(3) For a vertex $c \in B_1$, the number of neighbors of $c$ in $B_2$ is equal
to the number of neighbors of $\phi(c)$ in $B_2^J$.

The script goes over all possible configuration graphs. For each configuration
graph $J$, the script builds a vector whose branching number is an upper bound
on the branching number of Rule~(B\ref{rule:P3}) when it is applied on a graph
$G$  which matches $J$.

For each configuration graph $J$ the script generates a branching vector $R$
as follows.
Suppose that $G$ is a graph that matches $J$.
The script generates graphs of the form $J_{i_1\cdots i_p}$ like the generation
of first stage instances in Rule~(B\ref{rule:P3}),
except that now the process ends when $j^{J_{i_1\cdots i_p}} > 0$.
A graph $J_{i_1\cdots i_p}$ with $j^{J_{i_1\cdots i_p}} > 0$ generated by the
script will be called a \emph{first stage configuration graph}.
Consider some first stage configuration graph $J' = J_{i_1\cdots i_p}$.
This graph corresponds to the instance
$(G',k') = (G_{i_1\cdots i_p},k_{i_1\cdots i_p})$ that Rule~(B\ref{rule:P3})
generates when it is applied on the graph $G$ and the induced path $u,v,w$.
Note that if $j^{J'} \neq \infty$ then $(G',k')$ is not a first stage instance
and Rule~(B\ref{rule:P3}) will continue generating instances from $(G',k')$.
By the analysis of the case $j \geq 1$ above,
at the worst case, the number of first stage instances generated from $(G',k')$ 
is $2^{|E_1^{G'}|} = 2^{|E_1^{J'}|}$.
The vector $R(|E_1^{J'}|)$ gives a lower bound on the differences between
the parameter $k_{i_1\cdots i_p}$ and the parameters of these
$2^{|E_1^{G'}|}$ instances.
Therefore, the vector $R$ can be the concatenation of 
$R(|E_1^{J_{i_1\cdots i_p}}|)$ for every first stage graph $J_{i_1\cdots i_p}$.
We can get a better branching vector by giving a better bound on
$|S^{G''}|$ for the graphs $G''$ of the second stage instances $(G'',k'')$ that
are generated by Rule~(B\ref{rule:P3}) from an instance
$(G',k') = (G_{i_1\cdots i_p},k_{i_1\cdots i_p})$, where $G'$ matches
the first stage configuration graph $J' = J_{i_1\cdots i_p}$.
For our purpose, we it was suffices to give a simple bound based on the
edges between $A^{J'}$ and $B^{J'}$.
For example, suppose that there are vertices $b'_1,b'_2 \in B^{J'}$ such that
$N(b'_1) \cap A^{J'} = \{u'\}$ and $N(b'_2) \cap A^{J'} = \{w'\}$.
Then, in $G$ there are vertices $b_1,b_2 \in B$ such that
$N(b_1) \cap A = \{u\}$ and $N(b_2) \cap A = \{w\}$.
The edges $(b_1,u)$ and $(b_2,w)$ remain in every graph $G''$ of a second stage
instance $(G'',k'')$ generated from $(G',k')$.
Therefore, each such graph contains two edge disjoint induced $P_3$s
($b_1,u,v$ and $b_2,w,v$).
Therefore, $|S^{G''}| \geq 2$.
We use similar lower bounds in case of other edges between $A^{J'}$ and
$B^{J'}$.

All the branching vectors generated by the script have branching numbers less
than 1.393.
Therefore, the branching number of Rule~(B\ref{rule:P3}) is less then 1.393.

Since all the branching rules of the algorithm have branching numbers less than
1.415, it follows that the running time of the algorithm is $O^*(1.415^k)$.

\section{New algorithm}

Our algorithm first applies Rule~(B\ref{rule:path}) and Rule~(B\ref{rule:14})
until these rule cannot be applied
(note that Rule~(B\ref{rule:c4}) is not applied).
Let $G$ be a graph in which these rules cannot be applied.
If $G$ does not contain induced $C_4$ then
the algorithm applies Rule~(B\ref{rule:P3}).
Otherwise, the algorithm applies a new branching rule,
denoted (B\ref{rule:P3new}), which is as follows.
The algorithm picks an induced $C_4$ $u,v,w,u',u$.
Note that $u,v,w$ is an induced $P_3$.
Let $B_u$ and $C$ be the sets defined in the previous section.
We have that $|B_u| \geq 1$ since $u' \in B_u$.
If $C$ is a clique then the algorithm applies Rule~(B\ref{rule:P3}) on $u,v,w$.
Otherwise, let $x,y \in C$ be two non-adjacent vertices.
Note that $u,x,w,y,u$ is an induced $C_4$.
The algorithm applies Rule~(B\ref{rule:c4}) on the cycle $u,x,w,y,u$.
This generates two instances:
$(G_1,k_1) = (G-\{(u,x),(w,y)\},k-2)$ and
$(G_2,k_2) = (G-\{(x,w),(y,u)\},k-2)$.

Consider an instance $(G_i,k_i)$ and note that $u,v,w$ is an induced $P_3$ in
$G_i$.
If Rule~(B\ref{rule:14}) is applicable on $(G_i,k_i)$, the algorithm
applies Rule~(B\ref{rule:14}).
Now suppose that Rule~(B\ref{rule:14}) is not applicable.
In $G_i$, the vertex $x$ is adjacent to $v$ and not adjacent to $u$.
Therefore, $x \in B_u^{G_i}$.
We also have $u' \in B_u^{G_i}$, so by Lemma~\ref{lem:Bu} we have that
$B_u^{G_i} = \{u',x\}$.
Additionally, $x,y \notin C^{G_i}$.
If $C^{G_i}$ is a clique then the algorithm applies Rule~(B\ref{rule:P3})
on the graph $G_i$ and the path $u,v,w$.
Now suppose that $C^{G_i}$ is not a clique.
let $x_i,y_i \in C^{G_i}$ be two non-adjacent vertices.
We have that $u,x_i,w,y_i,u$ is an induced $C_4$ in $G_i$.
The algorithm applies Rule~(B\ref{rule:c4}) on the cycle $u,x_i,w,y_i,u$.
This generates two instances:
$(G_{i1},k_{i1}) = (G_i-\{(u,x_i),(w,y_i)\},k_i-2)$ and
$(G_{i2},k_{i2}) = (G_i-\{(x_i,w),(y_i,u)\},k_i-2)$.
Consider the instance $G_{i1}$.
Again, we have that $u,v,w$ is an induced $P_3$ in $G_{i1}$.
We now have that $B_u^{G_{i1}} = \{u',x,x_i\}$.
By Lemma~\ref{lem:Bu}, Rule~(B\ref{rule:14}) is applicable on $(G_{i1},k_{i1})$.
Using the same arguments, Rule~(B\ref{rule:14}) is applicable on
$(G_{i2},k_{i2})$.
Thus, the algorithm applies Rule~(B\ref{rule:14}) on $(G_{i1},k_{i1})$
and on $(G_{i2},k_{i2})$.

We now analyze the branching number of Rule~(B\ref{rule:P3new}).
There are three cases that we need to consider.
In the first case, the algorithm generates four instances
$G_{11}$, $G_{12}$, $G_{21}$, and $G_{22}$, and then applies
Rule~(B\ref{rule:14}) on each of these instances.
Therefore, the branching vector in this case is $(5,8,5,8,5,8,5,8)$
and the branching number is less than 1.404.
In the second case, the algorithm generates, without loss of generality,
the instances $G_{11}$, $G_{12}$, and $G_2$.
The algorithm then applies Rule~(B\ref{rule:14}) on $G_{11}$ and $G_{12}$, and
applies Rule~(B\ref{rule:14}) or Rule~(B\ref{rule:P3}) on $G_2$.
The worst case is when Rule~(B\ref{rule:14}) is applied on $G_2$.
The branching vector in this case is $(5,8,5,8,3,6)$,
and the branching number is less than 1.402.
In the third case, the algorithm generates the instances $G_1$ and $G_2$ and
applies Rule~(B\ref{rule:14}) or Rule~(B\ref{rule:P3}) on each of these
instances.
The worst branching vector in this case is $(3,6,3,6)$ and the branching number
is less than 1.398.
Therefore, the branching number of Rule~(B\ref{rule:P3new}) is less than 1.404.

Since all the branching rules of the algorithm have branching numbers less than
1.404, it follows that the running time of the algorithm is $O^*(1.404^k)$.

\bibliographystyle{abbrv}
\bibliography{cluster}

\end{document}